\documentclass[submission,copyright,creativecommons]{eptcs}
\providecommand{\event}{PxTP 2017} 
\usepackage{breakurl}             
\usepackage{underscore}           
\usepackage{listings}
\usepackage{proof}
\usepackage{amsmath}
\usepackage{empheq}
\usepackage{array}
\usepackage{amssymb}
\usepackage{latexsym}
\usepackage{amsfonts}
\usepackage{bussproofs}
\EnableBpAbbreviations
\usepackage{stmaryrd}
\usepackage{xspace}
\usepackage{amsthm}
\usepackage{cite}
\usepackage[usenames,dvipsnames,svgnames,table]{xcolor}
\usepackage[dvipsnames]{xcolor}
\newtheoremstyle{break}
  {\topsep}{\topsep}%
  {\itshape}{}%
  {\bfseries}{}%
  {\newline}{}%

\usepackage{marvosym}
\usepackage{yfonts}

\DeclareSymbolFont{letters}{OML}{txmi}{m}{it}

\newlength{\xleftmargin}
\def\T{\mathcal{T}}

\newcommand{\FV}{\text{FV}}

\newcommand\vvthinspace{\kern+0.041667em}
\newcommand\vthinspace{\kern+0.083333em}
\newcommand\negvthinspace{\kern-0.083333em}
\newcommand\negvvthinspace{\kern-0.041667em}

\DeclareFontFamily{OT1}{pzc}{}
\DeclareFontShape{OT1}{pzc}{m}{it}{<-> s * [1.10] pzcmi7t}{}
\DeclareMathAlphabet{\mathcalx}{OT1}{pzc}{m}{it}

\EnableBpAbbreviations

\newcommand\judgei[1]{{\vartriangleright}\;#1}
\newcommand\judge[2]{#1\;\vthinspace{\vartriangleright}\;#2}

\newcommand\axiomstrut{\vrule height .9ex depth 0pt width 0pt} 

\newcommand\metafun[1]{\ensuremath{\mathit{#1}}}

\newcommand\sym[1]{\mathsf{#1}}

\newcommand\substname{\metafun{subst}}

\newcommand\reify{\metafun{reify}}

\newcommand\subst[1]{\substname(#1)}
\newcommand\substapp[2]{#1(#2)}

\renewcommand\FV[1]{\metafun{FV}(#1)}

\newcommand\emptyctx{\varnothing}

\newcommand{\bothpol}{\mathclap{\raise.1ex\hbox{\hspace{2.4mm}\scalebox{.9}{$\top$}}}{\scalebox{.9}{$\bot$}}}

\newcommand\GAMMA{\mathrm{\Gamma}}

\DeclareFontFamily{OT1}{pzc}{}
\DeclareFontShape{OT1}{pzc}{m}{it}{<-> s * [1.10] pzcmi7t}{}
\DeclareMathAlphabet{\mathcalx}{OT1}{pzc}{m}{it}

\newcommand\medrightarrow{\mathrel{{{\color{black}\relbar}\kern-0.9ex\rlap{\color{white}\ensuremath{\blacksquare}}\kern-0.9ex}\joinrel{\color{black}\rightarrow}}}
\newcommand\Medrightarrow{\mathrel{{{\color{black}\Relbar}\kern-0.9ex\rlap{\color{white}\ensuremath{\blacksquare}}\kern-0.9ex}\joinrel{\color{black}\Rightarrow}}}
\newcommand\Medleftrightarrow{\mathrel{\Leftarrow\kern-1.685ex\Rightarrow}}

\newcommand\Refl{\textsc{Refl}}
\newcommand\Taut{\textsc{Taut}}
\newcommand\Tautof[1]{\Taut\smash{\ensuremath{_{#1}}}}
\newcommand\TautT{\Tautof{\T}}
\newcommand\Trans{\textsc{Trans}}
\newcommand\Cong{\textsc{Cong}}

\newcommand\Beta{\textsc{Beta}}
\newcommand\Bind{\textsc{Bind}}
\newcommand\Sko{\textsc{Sko}}
\newcommand\SkoEx{\Sko\smash{\ensuremath{_{\,\exists}}}}
\newcommand\SkoAll{\Sko\smash{\ensuremath{_{\,\forall}}}}
\newcommand\Let{\textsc{Let}}

\newcommand\BOX[1]{%
  {\vthinspace\setlength\fboxsep{.2ex}\fbox{\kern.2ex\ensuremath{#1\vphantom{hj}}\kern.2ex}\vthinspace}%
}

\theoremstyle{break}

\theoremstyle{definition}
\theoremstyle{remark}
\theoremstyle{lemma}
\theoremstyle{example}
\newtheorem{example}{Example}

\newtheorem{lemma}{Lemma}

\newcommand{\alt}{\ensuremath{~\mid~}}

\newcommand{\nter}[2][]{{\textcolor{OliveGreen}{\ensuremath{\langle}\emph{\sffamily{#2}}\ensuremath{\rangle^{#1}}}}\xspace}

\newcommand{\expr}[1]{{\textcolor{NavyBlue}{\texttt{#1}}}}
\newcommand{\ter}[1]{\expr{#1}\xspace}

\newcommand\HOLyHammer{HOL\kern-.075ex\raise.19ex\hbox{\textsc{y}}\kern-.05ex\-Hammer}

\makeatletter
\newcommand{\thickbar}{\mathpalette\@thickbar}
\newcommand{\@thickbar}[2]{{#1\mkern1.5mu\vbox{
  \sbox\z@{$#1\mkern-1.5mu#2\mkern-1.5mu$}%
  \sbox\tw@{$#1\overline{#2}$}%
  \dimen@=\dimexpr\ht\tw@-\ht\z@-.8\p@\relax
  \hrule\@height.7\p@ 
  \vskip\dimen@
  \box\z@}\mkern1.5mu}
}
\makeatother

\title{Language and Proofs for Higher-Order SMT\\ (Work in Progress)%
\footnote{This PxTP submission should be considered an addendum of the PxTP
presentation-only submission ``Scalable Fine-Grained Proofs for Formula Processing.''
Although we tried to make the current paper self-contained,
a knowledge of this previous work is useful to get the whole picture.}} 
\author{Haniel Barbosa\qquad Jasmin Christian Blanchette\qquad Simon Cruanes\\\\ Daniel El Ouraoui\qquad Pascal Fontaine\thanks{This work has
    been partially supported by the ANR/DFG project STU 483/2-1 SMArT
    ANR-13-IS02-0001 of the Agence Nationale de la Recherche, by the
    H2020-FETOPEN-2016-2017-CSA project SC$^\mathsf{2}$ (712689), and by the
    European Research Council (ERC) starting grant Matryoshka (713999).}
\\
  \institute{
    University of Lorraine, CNRS, Inria, and LORIA, Nancy, France \\
    Vrije Universiteit Amsterdam, Amsterdam, The Netherlands \\
     Max-Planck-Institut f\"ur Informatik, Saarbr\"ucken, Germany
  }
    \email{\{haniel.barbosa,jasmin.blanchette,simon.cruanes,daniel.el-ouraoui,pascal.fontaine\}@inria.fr}\\
}

\def\titlerunning{Language and Proofs for Higher-Order SMT (WIP)}
\def\authorrunning{Barbosa, Blanchette, Cruanes, El Ouraoui and  Fontaine}

\begin{document}
\maketitle

\begin{abstract}
  Satisfiability modulo theories (SMT) solvers have throughout the years been
  able to cope with increasingly expressive formulas, from ground logics to full
  first-order logic modulo theories.
  Nevertheless, higher-order logic within SMT is still little explored.  One main goal
  of the Matryoshka project, which started in March 2017, is to extend
  the reasoning capabilities of SMT solvers and other automatic provers beyond first-order
  logic.
  In this 
  preliminary report, we report on an extension of the SMT-LIB
  language, the standard input format of SMT solvers, to handle higher-order
  constructs.
  %
  We also discuss how to augment the proof format of the
  SMT solver veriT to accommodate these new constructs and the solving
  techniques they require.
\end{abstract}

\newcommand\cdclt{CDCL($\mathcal{T}$)}

\section{Introduction}

Higher-order (HO) logic is a pervasive setting for reasoning about numerous
real-world applications. In particular, it is widely used in proof assistants
(also known as interactive theorem provers) to provide trustworthy,
machine-checkable formal proofs of theorems.
A major challenge in these applications is to automate as much as
possible the production of these formal proofs, thereby reducing the burden
of proof on the users.

An effective approach for stronger automation is to rely on less
expressive but more automatic theorem provers to discharge some of the proof
obligations.
Systems such as \HOLyHammer, Miz$\mathbb{AR}$, Sledgehammer, and Why3,
which provide a
one-click connection from proof assistants to first-order
provers, have led in recent years to
considerable improvements in proof assistant automation
\cite{blanchette-et-al-2016-qed}.
Today, the leading automatic provers for first-order classical logic 
are based either on the superposition calculus~\cite{Bachmair1994,Nieuwenhuis2001-har} or
on \cdclt{}~\cite{Nieuwenhuis2006}. Those based on the latter are usually
called satisfiability modulo theory (SMT) solvers and are the focus of this
paper.

Our goal, as part of the Matryoshka
project,\footnote{\url{http://matryoshka.gforge.inria.fr/}} is to extend SMT
solvers to natively handle higher-order problems, thus avoiding the completeness
and performance issues associated with clumsy encodings.
%
In this paper, we present our first steps towards two contributions within our
established goal:\ to extend the input (problems) and output
(proofs) of SMT solvers to support higher-order constructs. Most SMT solvers
support SMT-LIB~\cite{Barrett2015} as an input format.
We report on a syntax extension for augmenting SMT-LIB with higher-order
functions with partial applications, $\lambda$-abstractions, and quantification
on higher-order variables (Section~\ref{sec:syntax-ext}).
Regrettably, there is no standard yet for proof output;
each proof-producing solver has its own format.
We focus on the proof format of the SMT solver
veriT~\cite{bouton2009verit}. This solver is known for its very detailed
proofs~\cite{Besson2011,Barbosa2017-proofs}, which are reconstructed in the
proof assistants
Isabelle\slash HOL
\cite{Blanchette2016-isar} and the GAPT system \cite{Ebner2016}.

Proofs in veriT accommodate the formula processing and the proof search
performed by the solver. Processing steps are represented using an extensible
set of inference rules described by Barbosa et al.~\cite{Barbosa2017-proofs}.
Here, we extend this calculus to support transformations such as $\beta$-reduction
and congruence with $\lambda$-abstractions, which are required by the new
constructs that can appear in higher-order problems
(Section~\ref{sec:proofs-ext}).

The \cdclt{} reasoning performed by veriT is represented by a resolution
proof, which consists of the resolution steps performed by the underlying SAT
solver and the lemmas added by the theory solvers and the instantiation
module. These steps are described in Besson et al.~\cite{Besson2011}.
The part of the proof corresponding to the actual proving will change according
to how we solve higher-order problems. In keeping with the \cdclt{} setting, the
reasoning is performed in a stratified manner. Currently, the SAT solver handles
the propositional reasoning, a combination of theory solvers tackle the ground
(variable-free) reasoning, and an instantiation module takes care of the first-order
reasoning. Our initial plan is to 
adapt the instantiation module so
that it can heuristically instantiate quantifiers with functional variables and
to extend veriT's underlying modular engine for
computing substitutions \cite{Barbosa2017}.
Since only modifications to the instantiation module are planned, the only rules
that must be adapted are those concerned with quantifier instantiation:

\begin{center}
    \AXC{$\phantom{\cdot}$}
    \RL{{\relax \textsc{Inst}$_\forall$}}
    \UIC{\strut$\forall x.\>\varphi[x]\rightarrow \varphi[t]$}
    \DP
    \qquad
    \AXC{$\phantom{\cdot}$}
    \RL{{\relax \textsc{Inst}$_\exists$}}
    \UIC{\strut$\varphi[t]\rightarrow\exists x.\>\varphi[x]$}
    \DP
\end{center}

\noindent
These rules are generic enough to be suitable also for higher-order
instantiation.  Here, we focus on adapting the rules necessary to suit the
new higher-order constructs in the formula processing steps.

\section{A Syntax Extension for the SMT-LIB Language}
\label{sec:syntax-ext}

By the time of starting this writing, the SMT-LIB standard was at version
2.5~\cite{Barrett2015}, and version 2.6 was in preparation.  Although some
discussions to extend the SMT-LIB language to higher-order logic have occurred in the
past, notably to include $\lambda$-abstractions, the format is currently based
on many-sorted first-order logic.  We here report on an extension of the
language in a pragmatic way to accommodate higher-order
constructs:\ higher-order functions with partial applications,
$\lambda$-abstractions, and quantifiers ranging over higher-order variables.
This extension is inspired by the work on TIP (Tools for Inductive
Provers)~\cite{rosen2015tip}, which is another pragmatic extension of SMT-LIB.

SMT-LIB contains commands to define atomic sorts and functions, but no
functional sorts.  The language is first extended 
so that functional
sorts can be built:

\begin{center}
\begin{tabular}{lcl}
 \nter{sort} & ::= & \phantom{\alt\ }\nter{identifier}
 \alt\ \ter{(} \nter{identifier} \nter[+]{sort} \ter{)}\\[2pt]
 &  & \alt\ \ter{(} \ter{->} \nter[+]{sort} \nter{sort} \ter{)}
\end{tabular}
\end{center}

\noindent
The second line is the addition to the original grammar. We use $\ter{(}\,
\ter{->}\, \nter[+]{sort}\, \nter{sort}\, \ter{)}$ rather than a special case of
$\ter{(}\,\nter[+]{identifier}\, \nter{sort}\, \ter{)}$ to avoid
ambiguities with parametric sorts and to have the same notation as the one
generally used for functional sorts.

The next modification is in the grammar for terms, which essentially adds a rule for
$\lambda$-abstractions and generalizes the
application so that any term can be applied to other terms:

\begin{center}
\begin{tabular}{lcl}
 \nter{term} & ::= & \nter{spec\_constant} \\
             &\alt&  \nter{qual\_identifier} \\
             &\alt & \ter{(} \nter{term} \nter[+]{term} \ter{)} \\
             &\alt & \ter{(} \ter{lambda} \ter{(} \nter[+]{sorted\_var} \ter{)}
                             \nter[+]{term} \ter{)} \\
             &\alt & \ter{(} \ter{let}
                             \ter{(} \nter[+]{var\_binding} \ter{)} \nter{term}
                     \ter{)} \\
             &\alt & \ter{(} \ter{forall}
                             \ter{(} \nter[+]{sorted\_var} \ter{)} \nter{term}
                     \ter{)} \\
             &\alt & \ter{(} \ter{exists}
                             \ter{(} \nter[+]{sorted\_var} \ter{)} \nter{term}
                     \ter{)} \\
             &\alt & \ter{(} \ter{match} \nter{term}
                             \ter{(} \nter[+]{match\_case} \ter{)}
                     \ter{)} \\
             &\alt & \ter{(} \ter{!} \nter{term} \nter[+]{attribute} \ter{)}\\[1ex]
 \nter{sorted\_var} & ::= & \ter{(} \nter{symbol} \nter{sort} \ter{)}
\end{tabular}
\end{center}

\noindent
The old rule $\ter{(} \nter{qual\_identifier}\ \nter[+]{term} \ter{)}$
is now redundant. Higher-order quantification requires no new syntax, since
sorts have been extended to accommodate functions.



Semantically, the well-sortedness rules in SMT-LIB are extended with
the following typing rules for the arrow constructor \verb+->+ and $\lambda$-abstraction:

\begin{center}
\rootAtTop
\AXC{\strut$\Sigma \vdash \lambda x.\; t : \sigma \rightarrow \tau$ }
\RL{\textsc{lambda}}
\UIC{\strut $\Sigma[x : \sigma]  \vdash t : \tau$}
\DP
\qquad
\rootAtTop
\AXC{\strut $\Sigma \vdash u\; v : \tau$ }
\RL{\textsc{app}}
\UIC{\strut $\Sigma \vdash u : \sigma \rightarrow \tau \; \; \; \Sigma \vdash v : \sigma$}
\DP
\end{center}
Where a judgment is composed of two items. On the left hand side, a signature
$\Sigma$, which is a tuple of function and constant symbols.
On the right hand side, a term annotated by its type.
The notation $\Sigma[x:\tau]$ stands for the signature that maps $x$ to the
type $\tau$.


If we want to define a function taking an integer as argument and
returning a function from integers to integers, it is now possible to write
\,\verb+(declare-fun f (Int) (-> Int Int))+.\,
The following example illustrates
higher-order functions, terms representing a function, and partial applications:

\begin{center}
  \begin{lstlisting}[frame=none,basicstyle=\ttfamily, mathescape, xleftmargin=\xleftmargin]
(set-logic UFLIA)
(declare-fun g (Int) (-> Int Int))
(declare-fun h (Int Int) Int)
(declare-fun f ((-> Int Int)) Int)
(assert (= (f (h 1)) ((g 1) 2)))
(exit)
  \end{lstlisting}
\end{center}

\noindent
The term \verb+(g 1)+ is a function from \verb+Int+ to \verb+Int+, in agreement
with the sort of \verb+g+.  Then it is applied to \verb+2+ in the expression
\verb+((g 1) 2)+ of sort \verb+Int+.  The term \verb+(h 1)+ is a partial
application of the binary function \verb+h+, and is thus a unary function.  The
term \verb+(f (h 1))+ is therefore well typed and is an \verb+Int+.
Note that in our presentation all functions of type \verb+(-> Int Int +\dots\verb+ Int)+
are equivalent to \verb+(-> Int (-> Int (-> +\dots\verb+ Int)))+. This implies, in
particular, that in the example above \verb+((g 1) 2)+ is semantically
equal to \verb+(g 1 2)+. 
More precisely we may considerate the three different declaration of f below:

\begin{center}
  \begin{lstlisting}[frame=none,basicstyle=\ttfamily, mathescape, xleftmargin=\xleftmargin]
(declare-fun f () (-> Int (-> Int Int)))
(declare-fun f (Int) (-> Int Int))
(declare-fun f (Int Int) Int)
  \end{lstlisting}
\end{center}
as the unique form \verb+(declare-fun f () (-> Int Int Int))+. This follows from
\verb+->+ being right associative. The next example features $\lambda$-abstraction:

\begin{center}
  \begin{lstlisting}[frame=none,basicstyle=\ttfamily, mathescape, xleftmargin=\xleftmargin]
(set-logic UFLIA)
(declare-fun g (Int) (Int))
(assert
  (= ((lambda ((f (-> Int Int)) (x Int)) f x) g 1) (g 1)))
(exit)
  \end{lstlisting}
\end{center}

\noindent
The term \verb+(lambda ((f (-> Int Int)) (x Int)) f x)+ is an anonymous function
that takes a function \verb+f+ and an integer \verb+x+ as arguments.
It is applied to \verb#g# and \verb+1+, and the fully applied
term is stated to be equal to \verb+(g 1)+. The assertion is a tautology
(thanks to $\beta$-reduction).

\section{An Extension for the veriT Proof Format}
\label{sec:proofs-ext}

Our setting is classical higher-order logic as defined by the extended SMT-LIB
language above, or abstractly described by the following grammar:
%
\begin{equation*}
  M \; ::= \;\; x \;\; |
  \; \;  c  \;\;
  |  \;\; M \; M \;\;
  |  \;\; \lambda x.\> M \;\;
  | \;\; \mathrm{let} \; \bar x_n \simeq \thickbar{M}_n \; \mathrm{in} \; M
\end{equation*}
%
\noindent where formulas are terms of Boolean type.  We rely on the metatheory
defined by Barbosa et al.~\cite{Barbosa2017-proofs}. Besides the axioms for
characterizing Hilbert choice and `let' described there, we add the following
axiom for $\lambda$-abstraction, where $\simeq$ denotes the equality predicate:
\[
  \tag{$\beta$}
  \models  (\lambda x.\; t[x])\> s \; \simeq \; t[s]
\]
In general, the notation $t[\bar x_n]$ stands for a term that may depend on distinct
variables $\bar x_n$; $t[\bar s_n]$ is the corresponding term where the terms
$\bar s_n$ are simultaneously substituted for $\bar x_n$; bound variables in $t$
are renamed to avoid capture.
For readability, and because it is natural with a higher-order calculus, we
present the rules in curried form---that is, functions can be partially applied,
and rules must only consider unary functions.

The notion of context is as in Barbosa et al.:
\[
    \Gamma \; ::= \;\; \varnothing \;\; |
    \; \; \Gamma,\,x \;\;
    |  \;\; \Gamma,\, \bar x_n \mapsto \bar t_n
\]
Each context entry either \emph{fixes} a variable~$x$ or defines a
\emph{substitution} $\{\bar x_n \mapsto \bar t_n\}$.
%
%
Abstractly, a context $\GAMMA$ fixes a set of variables and specifies a
substitution $\subst{\GAMMA}$. The substitution is the identity for $\emptyctx$
and is defined as follows in the other cases:
\begin{align*}
  \subst{\GAMMA,\, x} & = \subst{\GAMMA}[x \mapsto x]
  \, 
&
  \subst{\GAMMA,\, \bar x_n \mapsto \bar t_n} & = \subst{\GAMMA} \mathrel\circ \{\bar x_n \mapsto \bar t_n\}
\end{align*}
In the first equation, the $[x \mapsto x]$ update shadows any replacement of $x$
induced by $\GAMMA$.
We write $\substapp{\GAMMA}{t}$ to abbreviate the capture-avoiding substitution
$\subst{\GAMMA}(t)$.

Our new set of rules is similar to that in Barbosa et al.  The rules \Trans,\,
\SkoEx,\, \SkoAll,\, \Let,\, and \TautT{} are unchanged.  The \Bind{} rule is
modified to accommodate the new $\lambda$-binder:

\begin{center}
    \AXC{\strut$\judge{\GAMMA,\,y,\,x\mapsto y}{s\simeq t}$}
    \RL{{\relax \Bind}\quad\mbox{if\, $y \notin \FV{Bx.\>s}$}}
    \UIC{\strut$\judge{\GAMMA}{(Bx.\>s)\simeq (By.\>t)}$}
    \DP
\end{center}

\noindent
The metavariable $B$ ranges over $\forall$, $\exists$, and $\lambda$.  The
\Cong{} rule is also modified to accommodate new cases.  With respect to the
first-order calculus, the left-hand side of an application can be an arbitrarily
complex term, and not simply a function or predicate symbol.
Rewriting can now occur also on these complex terms.
The updated \Cong{} rule is as follows:

\begin{center}
   \AXC{\strut$\judge{\GAMMA}{s \simeq s'}{}$}
        \AXC{\strut$\judge{\GAMMA}{t \simeq t'}$}
    \RL{{\relax \Cong}}
    \BIC{\strut$\judge{\GAMMA}{s\> t\simeq s'\> t'}$}
    \DP
\end{center}

The only genuinely new rule is for $\beta$-reduction---that is,
the substitution of an argument in the body of a $\lambda$-abstraction.
It is similar in form to the \Let{} rule from
the first-order calculus:

\begin{center}
    \AXC{\strut$\judge{\GAMMA}{v \simeq s}{}$}
      \AXC{\strut$\judge{\GAMMA,\, x \mapsto s}{t\simeq u}$}
    \RL{{\relax \Beta}\quad\mbox{if\,
    $\substapp{\GAMMA}{s} = s $}}
    \BIC{\strut$\judge{\GAMMA}{(\lambda x.\> t)\; v \simeq u}$}
    \DP
\end{center}

\noindent
Indeed, $(\mathrm{let}~ x\simeq u ~\mathrm{in}~ t)$ and $(\lambda x.\> t)\, u$
are semantically equal.

\begin{example}\rm
  The derivation tree of the normalization of $(\lambda x.\; \sym{p} \; x\; x)
  \; \sym{a}$ is as follows:
    \[
    \AXC{\axiomstrut}
    \RL{{\small $\Cong$}}
    \UIC{\strut$\judgei{\sym{a} \simeq \sym{a}}$}
      \AXC{\axiomstrut}
      \RL{{\small \Refl}}
      \UIC{\strut$\judge{x\mapsto\sym{a}}{\sym{p}\simeq \sym{p}}$}
        \AXC{\axiomstrut}
        \RL{{\small \Refl}}
        \UIC{\strut$\judge{x\mapsto\sym{a}}{x\simeq \sym{a}}$}
      \RL{{\small \Cong}}
      \BIC{\strut$\judge{x\mapsto\sym{a}}{\sym{p}\,x\simeq \sym{p}\,\sym{a}}$}
        \AXC{\axiomstrut}
        \RL{{\small \Refl}}
        \UIC{\strut$\judge{x\mapsto\sym{a}}{x\simeq \sym{a}}$}
      \RL{{\small \Cong}}
      \BIC{\strut$\judge{x\mapsto\sym{a}}{\sym{p}\,x\, x\simeq \sym{p}\,\sym{a}\, \sym{a}}$}
    \RL{{\small \Beta}}
    \BIC{\strut$\judgei{(\lambda x.\> \sym{p}\,x\, x)\, \sym{a}\,\simeq \sym{p}\,\sym{a}\,\sym{a}}$}
    \DP
  \]
\end{example}
\begin{example}\rm
The following tree features a $\beta$-redex under a $\lambda$-abstraction.
Let $\, \Gamma_1 = w,\, x \mapsto w$; $\, \Gamma_2 = \Gamma_1,\, y \mapsto \sym{f}\> w$; and
$\, \Gamma_3 = \Gamma_2,\, z \mapsto \sym{f}\> w$:
  \[
    \AXC{\axiomstrut}
    \RL{{\small \Refl}}
    \UIC{\strut$\judge{\GAMMA_1}{\sym{f}\simeq \sym{f}}$}
      \AXC{\axiomstrut}
      \RL{{\small \Refl}}
      \UIC{\strut$\judge{\GAMMA_1}{x\simeq w}$}
    \RL{{\small \Cong}}
    \BIC{\strut$\judge{\GAMMA_1}{\sym{f}\> x\simeq \sym{f}\> w}$}
      \AXC{\axiomstrut}
      \RL{{\small \Refl}}
      \UIC{\strut$\judge{\GAMMA_2}{y\simeq \sym{f}\>w}$}
        \AXC{\axiomstrut}
        \RL{{\small \Refl}}
        \UIC{\strut$\judge{\GAMMA_3}{\sym{p}\>z\simeq \sym{p}\>(\sym{f}\>w)}$}
      \RL{{\small \Beta}}
      \BIC{\strut$\judge{\GAMMA_2}{(\lambda z.\> \sym{p}\> z)\> y\simeq \sym{p} \> (\sym{f}\> w)}$}
    \RL{{\small \Beta}}
    \BIC{\strut$\judge{\GAMMA_1}{(\lambda y.\> (\lambda z.\> \sym{p}\> z) \>
        y)\> (\sym{f}\> x)\simeq \sym{p} \> (\sym{f}\> w)}$}
    \RL{{\small \Bind}}
    \UIC{\strut$\judgei{(\lambda x.\>
  (\lambda y.\> (\lambda z.\> \sym{p}\> z) \> y)\> (\sym{f}\> x))\simeq (\lambda w.\> \sym{p} \> (\sym{f}\> w))}$}
    \DP
\]
\end{example}
\begin{example}\rm
The transitivity rule is useful when the applied term reduces to a $\lambda$-abstraction.  Let
$\,\Gamma_1 =  w,\, y \mapsto w$;
$\,\Gamma_2 = \Gamma_1,\, x \mapsto w$;
$\,\Gamma_3 = \Gamma_1,\, w_1 \mapsto \sym{p} \; w$;
$\,\Gamma_4 = \Gamma_1,\, x \mapsto \lambda w_1.\> w$;
$\,\Gamma_5 = \Gamma_1,\, w_1 ,\,x \mapsto w_1 $; and
$\,\Gamma_6 = \Gamma_4,\, z \mapsto \lambda w_1.\> w$:
\[
   \AXC{\axiomstrut}
   \RL{{\footnotesize \Refl}}
   \UIC{\strut\small$\judge{ \Gamma_1 }{y \simeq w }$}
   \AXC{\axiomstrut}
   \RL{{\footnotesize\Refl}}
   \UIC{\strut\small$\judge{ \Gamma_2 }{\sym{p}\; x \simeq \sym{p}\; w }$}
  \RL{{\footnotesize \Beta}}
  \BIC{\strut\small$\judge{\Gamma_1}{(\lambda x.\> \sym{p}\; x)\; y \simeq \sym{p}\; w}$}
  \AXC{\axiomstrut}
  \noLine
  \UIC{\strut\small $\Pi$}
 \RL{{\footnotesize \Cong}}
 \BIC{\strut\small$\judge{\Gamma_1 }{((\lambda x.\> (\lambda z.\> z)\; x)\; (\lambda x.\> y)) \;((\lambda x.\> \sym{p}\; x)\; y) \simeq (\lambda w_1.\> w)\; (\sym{p} \; w)}$}
  \AXC{\axiomstrut}
  \RL{{\footnotesize \Refl}}
  \UIC{\strut\small$\judge{\Gamma_1}{\sym{p} \; w \simeq \sym{p}\;w}$}
  \AXC{\axiomstrut}
  \RL{{\scriptsize \Cong}}
  \UIC{\strut\small$\judge{\Gamma_3 }{ w \simeq w }$}
 \RL{{\footnotesize \Beta}}
 \BIC{\strut\small$\judge{\Gamma_1}{(\lambda w_1.\> w)\; (\sym{p} \; w) \simeq w}$}
\RL{{\footnotesize \Trans}}
\BIC{\strut\small$\judge{\Gamma_1 }{((\lambda x.\> (\lambda z.\> z)\; x)\; (\lambda x.\> y)) \;((\lambda x.\> \sym{p}\; x)\; y)  \simeq w}$}
\RL{{\footnotesize \Bind}}
\UIC{\strut\small $\judgei{ (\lambda y.\>(\lambda x.\> (\lambda z.\> z)\; x)\;(\lambda x.\> y)) \;((\lambda x.\> \sym{p}\; x)\; y)) \simeq (\lambda w.\> w)}$}
\DP
\]
where $\Pi$ stands for the subtree\strut
\[
 \AXC{\axiomstrut}
 \RL{{\small \Refl}}
 \UIC{\strut$\judge{\Gamma_5}{y \simeq w}$}
 \RL{{\small \Bind}}
 \UIC{\strut$\judge{\Gamma_1}{(\lambda x.\> y) \simeq (\lambda w_1.\> w)}$}
  \AXC{\axiomstrut}
  \RL{{\small \Refl}}
  \UIC{\strut$\judge{\Gamma_4 }{x \simeq (\lambda w_1.\> w)}$}
  \AXC{\axiomstrut}
  \RL{{\small \Refl}}
  \UIC{\strut$\judge{\Gamma_6}{z \simeq (\lambda w_1.\> w)}$}
 \RL{{\small \Beta}}
 \BIC{\strut$\judge{\Gamma_4}{(\lambda z.\> z)\; x  \simeq (\lambda w_1.\> w)}$}
\RL{{\small \Beta}}
\BIC{\strut$\judge{ \Gamma_1 }{(\lambda x.\> (\lambda z.\> z)\; x)\; (\lambda x.\> y) \simeq (\lambda w_1.\> w)}$}
\DP
\]
\end{example}

The soundness of the extended calculus is a simple extension of the
soundness proof in the technical report by Barbosa et al.~\cite{Barbosa2017-proofs-extended}.  We focus on the
extensions.  Recall that the proof uses an encoding of terms and context in
$\lambda$-calculus, based on the following grammar:
\begin{equation*}
  M \; ::= \;\; \BOX{t} \;\; |
  \;\; (\lambda x.\> M) \; \;
  |  \;\; (\lambda \bar x_n.\> M) \; \bar t_n
\end{equation*}
As previously, $\reify(M \, \simeq\,N)$ is defined as $\forall
\bar x_n.\; t \, \simeq \, u$ if $M \, =_{\alpha \beta} \, \lambda x_1
\ldots \lambda x_n .\> \BOX{t}$ and $N \, =_{\alpha \beta} \, \lambda x_1 \ldots
\lambda x_n .\> \BOX{u}$.  The encoded rules are as follows:

\vskip\abovedisplayskip

\noindent\hbox{}\hfill
   \AXC{\strut${ M[s] \simeq N[s']}{}
    \quad
    {M[t] \simeq N[t']}
    $}
    \RL{{\relax \Cong}}
    \UIC{\strut${M[s \> t]  \simeq N[s' \> t']}$}
    \DP
    \qquad
    \AXC{\strut${M[\lambda y.\>(\lambda x.\> s)\; y] \simeq N[\lambda y.\> t]}$}
    \RL{{\relax \Bind}\quad\mbox{if\, $y \notin \FV{Bx.\>s}$}}
    \UIC{\strut${M[Bx.\>s]\simeq N[By.\>t]}$}
    \DP
\hfill\hbox{}

\vskip.7\abovedisplayskip

\noindent\hbox{}\hfill
\AXC{\strut${M[v] \simeq N[s]}{}
    \quad
    {M[(\lambda x .\> t)\; s] \simeq N[u] }$}
    \RL{{\relax \Beta}\quad\mbox{if\,
    $M[v] =_{\alpha \beta} N[s] $}}
    \UIC{\strut${M[(\lambda x.\> t)\; v] \simeq N[u]}$}
    \DP
\hfill\hbox{}

\vskip\belowdisplayskip

\begin{lemma}
If the judgment $M \simeq N$ is derivable using the encoded inference system with the
theories $\mathcal{T}_1 \ldots \mathcal{T}_n$, then $\models_\mathcal{T}$ $\reify(M\, \simeq \, N)$ with
$\mathcal{T}$ = $\mathcal{T}_1 \cup \cdots \cup \mathcal{T}_n$ $\cup$ $\simeq$
$\cup$ $\epsilon_1$ $\cup$ $\epsilon_2$ $\cup$ $\mathrm{let}$ $\cup$ $\beta$.
\end{lemma}
\begin{proof}
  The proof is by induction over the derivation $ M \simeq N$.  We only provide
  here the three new cases:\

  \smallskip

  \noindent
  \textsc{Case} \Bind{}$\; B=\lambda$: The induction hypothesis is
  $\models_\mathcal{T}$ $\reify(M[\lambda y.\>(\lambda x.\> s[x])\; y]
  \simeq N[\lambda y.\> t[y]])$.  Using ($\beta$) and the side condition of
  the rule, we can also deduce that $\models_\mathcal{T}$ $\reify(M[\lambda
  y.\>s[y])] \simeq N[\lambda y.\> t[y]])$. Hence by
  $\alpha$-conversion this is equivalent to ${\models_\mathcal{T}}\;\reify(M[\lambda
  x.\>s[x]] \simeq N[\lambda y.\> t[y]])$.

  \smallskip

  \noindent
  \textsc{Case} \Cong: This case follows directly from equality in a higher-order setting.

  \smallskip

  \noindent
  \textsc{Case} \Beta: This case follows directly from $(\beta)$ and equality in a
  higher-order setting.

  \smallskip

  \noindent
  The remaining cases are similar to Barbosa et al.
\end{proof}

The auxiliary functions $L(\Gamma)[t]$ and $R(\Gamma)[u]$ are used to
encode the judgment of the original inference system $\Gamma$ $\vartriangleright$ $t\; \simeq \; u$.
They are defined over the structure of the context, as follows:
\begin{align*}
  L(\varnothing)[t]\, &= \, \BOX{t}   &&R(\varnothing)[u]\, = \, \BOX{u} \\
  L(x,\,\Gamma)[t]\, &= \, \lambda x.\> L(\Gamma)[t]   &&R(x,\,\Gamma)[u]\, = \, \lambda.\> L(\Gamma)[u] \\
  L(\bar x_n \mapsto \bar s_n,\Gamma)[t]\, &= \, (\lambda \bar x_n.\> L(\Gamma)[t])\, \bar s_n   &&R(\bar x_n \mapsto \bar s_n,\Gamma)[u]\, = \, (\lambda \bar x_n.\> L(\Gamma)[u])\, \bar s_n
\end{align*}

\begin{lemma}
  If the judgment $\Gamma \vartriangleright t \simeq u$ is derivable using the original inference system, the equality
  $L(\Gamma)[t]$ $\simeq$ $R(\Gamma)[u]$ is derivable using the encoded inference system.
\end{lemma}
\begin{proof} The proof is by induction over the derivation $\Gamma \vartriangleright t \simeq u$, we give only the three new cases:\

  \smallskip

  \noindent
  \textsc{Case} \Bind{}  with $B=\lambda$:
  The encoded antecedent is
  $M[\lambda y.\>(\lambda x.\;s)\;y]
  \simeq N[\lambda y.\;t]$
  (i.e.,
  $L({\Gamma,\,\allowbreak y,\,\allowbreak x\mapsto\nobreak y}){[s]} \simeq R({\Gamma,\,y,\,x\mapsto y}){[t]}$),
  and the encoded succedent is
  $M[\lambda x.\>s] \simeq N[\lambda y.\>t]$.
  By the induction hypothesis, the encoded antecedent is derivable. Thus, by
  the encoded \Bind{} rule, the encoded succedent is derivable.

  \smallskip

  \noindent
  \textsc{Case} \Cong: Similar to \Bind.

  \smallskip

  \noindent
  \textsc{Case} \Beta: Similar to \Let{} with $n = 1$.

  \smallskip

  \noindent
  The remaining cases are similar to Barbosa et al.
\end{proof}

\begin{lemma}[\textbf{Soundness of Inferences}]
  If the judgment $\Gamma \vartriangleright t \simeq u$ is derivable using the original inference system with the
  theories $\mathcal{T}_1 \ldots \mathcal{T}_n$, then $\models_\mathcal{T}$ $\Gamma (t)$ $\simeq$ $u$ with
  $\mathcal{T}$ = $\mathcal{T}_1 \cup \ldots \cup \mathcal{T}_n$ $\cup$ $\simeq$
  $\cup$ $\epsilon_1$ $\cup$ $\epsilon_2$ $\cup$ $\mathrm{let}$ $\cup$ $\beta$.
\end{lemma}
\begin{proof} Using the above updated lemmas, the proof is identical to the one for the original calculus. 
\end{proof}

\section{Conclusion and Future Work}

We have presented a preliminary extension of the SMT-LIB syntax and of the
veriT proof format to support higher-order constructs in SMT problems and proofs.
Partial applications, $\lambda$-abstractions, and quantification over functional
variables can now be understood by a solver compliant with these languages.
The only relatively challenging element of these extensions so far concerns
the rules for representing detailed proofs of formula processing. The
next step is to extend the generic proof-producing formula processing
algorithm from Barbosa et al.~\cite{Barbosa2017-proofs}. Given the structural
similarity between the introduced extensions and the previous proof calculus,
we expect this to be straightforward.

A more interesting challenge will be to reconstruct these new proofs in proof
assistants, to allow full integration of a higher-order SMT solver. Since
detailed proofs are produced, with proof checking being guaranteed to have
reasonable complexity, we are confident to be able to produce effective
implementations.

With the foundations in place, the next step will be to implement the
automatic reasoning machinery for higher-order formulas and properly evaluating
its effectiveness.  Moreover, when providing support for techniques involving,
for example, inductive datatypes, we will need to augment the proof format
accordingly. 

\paragraph{Acknowledgment}

We would like to thank the anonymous reviewers for their
comments.  Between the initial version of this document and the current one, the
SMT-LIB extension has been greatly influenced by discussions with Clark Barrett
and Cesare Tinelli (the SMT-LIB managers, together with Pascal Fontaine) and
they should also be considered authors of this syntax extension.

\bibliographystyle{eptcs}
\bibliography{generic}

\begin{thebibliography}{10}
\providecommand{\bibitemdeclare}[2]{}
\providecommand{\surnamestart}{}
\providecommand{\surnameend}{}
\providecommand{\urlprefix}{Available at }
\providecommand{\url}[1]{\texttt{#1}}
\providecommand{\href}[2]{\texttt{#2}}
\providecommand{\urlalt}[2]{\href{#1}{#2}}
\providecommand{\doi}[1]{doi:\urlalt{http://dx.doi.org/#1}{#1}}
\providecommand{\bibinfo}[2]{#2}

\bibitemdeclare{article}{Bachmair1994}
\bibitem{Bachmair1994}
\bibinfo{author}{Leo \surnamestart Bachmair\surnameend} \&
  \bibinfo{author}{Harald \surnamestart Ganzinger\surnameend}
  (\bibinfo{year}{1994}): \emph{\bibinfo{title}{{Rewrite-Based Equational
  Theorem Proving with Selection and Simplification}}}.
\newblock {\sl \bibinfo{journal}{Journal of Logic and Computation}}
  \bibinfo{volume}{4}(\bibinfo{number}{3}), pp. \bibinfo{pages}{217--247},
  \doi{10.1093/logcom/4.3.217}.

\bibitemdeclare{incollection}{Barbosa2017-proofs}
\bibitem{Barbosa2017-proofs}
\bibinfo{author}{Haniel \surnamestart Barbosa\surnameend},
  \bibinfo{author}{Jasmin~Christian \surnamestart Blanchette\surnameend} \&
  \bibinfo{author}{Pascal \surnamestart Fontaine\surnameend}
  (\bibinfo{year}{2017}): \emph{\bibinfo{title}{Scalable Fine-Grained Proofs
  for Formula Processing}}.
\newblock In \bibinfo{editor}{Leonardo \surnamestart de~Moura\surnameend},
  editor: {\sl \bibinfo{booktitle}{Conference on Automated Deduction (CADE)}},
  {\sl \bibinfo{series}{LNCS}} \bibinfo{volume}{10395},
  \bibinfo{publisher}{Springer}, pp. \bibinfo{pages}{398--412},
  \doi{10.1007/978-3-319-63046-5_25}.

\bibitemdeclare{techreport}{Barbosa2017-proofs-extended}
\bibitem{Barbosa2017-proofs-extended}
\bibinfo{author}{Haniel \surnamestart Barbosa\surnameend},
  \bibinfo{author}{Jasmin~Christian \surnamestart Blanchette\surnameend} \&
  \bibinfo{author}{Pascal \surnamestart Fontaine\surnameend}
  (\bibinfo{year}{2017}): \emph{\bibinfo{title}{{Scalable Fine-Grained Proofs
  for Formula Processing}}}.
\newblock \bibinfo{type}{Research Report}, \bibinfo{institution}{Inria},
  \doi{10.1007/978-3-319-63046-5_25}.
\newblock \urlprefix\url{https://hal.inria.fr/hal-01526841}.

\bibitemdeclare{inproceedings}{Barbosa2017}
\bibitem{Barbosa2017}
\bibinfo{author}{Haniel \surnamestart Barbosa\surnameend},
  \bibinfo{author}{Pascal \surnamestart Fontaine\surnameend} \&
  \bibinfo{author}{Andrew \surnamestart Reynolds\surnameend}
  (\bibinfo{year}{2017}): \emph{\bibinfo{title}{Congruence Closure with Free
  Variables}}.
\newblock In \bibinfo{editor}{Axel \surnamestart Legay\surnameend} \&
  \bibinfo{editor}{Tiziana \surnamestart Margaria\surnameend}, editors: {\sl
  \bibinfo{booktitle}{Tools and Algorithms for Construction and Analysis of
  Systems (TACAS)}}, {\sl \bibinfo{series}{LNCS}} \bibinfo{volume}{10206}, pp.
  \bibinfo{pages}{214--230}, \doi{10.1007/978-3-662-54580-5_13}.

\bibitemdeclare{techreport}{Barrett2015}
\bibitem{Barrett2015}
\bibinfo{author}{Clark \surnamestart Barrett\surnameend},
  \bibinfo{author}{Pascal \surnamestart Fontaine\surnameend} \&
  \bibinfo{author}{Cesare \surnamestart Tinelli\surnameend}
  (\bibinfo{year}{2015}): \emph{\bibinfo{title}{{The SMT-LIB Standard: Version
  2.5}}}.
\newblock \bibinfo{type}{Technical Report}, \bibinfo{institution}{Department of
  Computer Science, The University of Iowa}.
\newblock \bibinfo{note}{Available at {\tt www.SMT-LIB.org}}.

\bibitemdeclare{inproceedings}{Besson2011}
\bibitem{Besson2011}
\bibinfo{author}{Fr{\'e}d{\'e}ric \surnamestart Besson\surnameend},
  \bibinfo{author}{Pascal \surnamestart Fontaine\surnameend} \&
  \bibinfo{author}{Laurent \surnamestart Th{\'e}ry\surnameend}
  (\bibinfo{year}{2011}): \emph{\bibinfo{title}{A Flexible Proof Format for
  {SMT}: a Proposal}}.
\newblock In \bibinfo{editor}{Pascal \surnamestart Fontaine\surnameend} \&
  \bibinfo{editor}{Aaron \surnamestart Stump\surnameend}, editors: {\sl
  \bibinfo{booktitle}{Workshop on Proof eXchange for Theorem Proving (PxTP)}}.
\newblock \urlprefix\url{https://hal.inria.fr/hal-00642544}.

\bibitemdeclare{article}{Blanchette2016-isar}
\bibitem{Blanchette2016-isar}
\bibinfo{author}{Jasmin~Christian \surnamestart Blanchette\surnameend},
  \bibinfo{author}{Sascha \surnamestart B{\"{o}}hme\surnameend},
  \bibinfo{author}{Mathias \surnamestart Fleury\surnameend},
  \bibinfo{author}{Steffen~Juilf \surnamestart Smolka\surnameend} \&
  \bibinfo{author}{Albert \surnamestart Steckermeier\surnameend}
  (\bibinfo{year}{2016}): \emph{\bibinfo{title}{Semi-intelligible {I}sar Proofs
  from Machine-Generated Proofs}}.
\newblock {\sl \bibinfo{journal}{Journal of Automated Reasoning}}
  \bibinfo{volume}{56}(\bibinfo{number}{2}), pp. \bibinfo{pages}{155--200},
  \doi{10.1007/s10817-015-9335-3}.

\bibitemdeclare{article}{blanchette-et-al-2016-qed}
\bibitem{blanchette-et-al-2016-qed}
\bibinfo{author}{Jasmin~Christian \surnamestart Blanchette\surnameend},
  \bibinfo{author}{Cezary \surnamestart Kaliszyk\surnameend},
  \bibinfo{author}{Lawrence~C. \surnamestart Paulson\surnameend} \&
  \bibinfo{author}{Josef \surnamestart Urban\surnameend}
  (\bibinfo{year}{2016}): \emph{\bibinfo{title}{Hammering towards {QED}}}.
\newblock {\sl \bibinfo{journal}{Journal of Formalized Reasoning}}
  \bibinfo{volume}{9}(\bibinfo{number}{1}), pp. \bibinfo{pages}{101--148},
  \doi{10.6092/issn.1972-5787/4593}.

\bibitemdeclare{inproceedings}{bouton2009verit}
\bibitem{bouton2009verit}
\bibinfo{author}{Thomas \surnamestart Bouton\surnameend},
  \bibinfo{author}{Diego Caminha~B. \surnamestart de~Oliveira\surnameend},
  \bibinfo{author}{David \surnamestart D{\'{e}}harbe\surnameend} \&
  \bibinfo{author}{Pascal \surnamestart Fontaine\surnameend}
  (\bibinfo{year}{2009}): \emph{\bibinfo{title}{{veriT}: {A}n {O}pen,
  {T}rustable and {E}fficient {SMT-S}olver}}.
\newblock In \bibinfo{editor}{Renate~A. \surnamestart Schmidt\surnameend},
  editor: {\sl \bibinfo{booktitle}{Conference on Automated Deduction (CADE)}},
  {\sl \bibinfo{series}{LNCS}} \bibinfo{volume}{5663},
  \bibinfo{publisher}{Springer}, pp. \bibinfo{pages}{151--156},
  \doi{10.1007/978-3-642-02959-2_12}.

\bibitemdeclare{inproceedings}{Ebner2016}
\bibitem{Ebner2016}
\bibinfo{author}{Gabriel \surnamestart Ebner\surnameend},
  \bibinfo{author}{Stefan \surnamestart Hetzl\surnameend},
  \bibinfo{author}{Giselle \surnamestart Reis\surnameend},
  \bibinfo{author}{Martin \surnamestart Riener\surnameend},
  \bibinfo{author}{Simon \surnamestart Wolfsteiner\surnameend} \&
  \bibinfo{author}{Sebastian \surnamestart Zivota\surnameend}
  (\bibinfo{year}{2016}): \emph{\bibinfo{title}{System Description: {GAPT}
  2.0}}.
\newblock In \bibinfo{editor}{Nicola \surnamestart Olivetti\surnameend} \&
  \bibinfo{editor}{Ashish \surnamestart Tiwari\surnameend}, editors: {\sl
  \bibinfo{booktitle}{International Joint Conference on Automated Reasoning
  (IJCAR)}}, {\sl \bibinfo{series}{LNCS}} \bibinfo{volume}{9706},
  \bibinfo{publisher}{Springer}, pp. \bibinfo{pages}{293--301},
  \doi{10.1007/978-3-319-40229-1_20}.

\bibitemdeclare{article}{Nieuwenhuis2006}
\bibitem{Nieuwenhuis2006}
\bibinfo{author}{Robert \surnamestart Nieuwenhuis\surnameend},
  \bibinfo{author}{Albert \surnamestart Oliveras\surnameend} \&
  \bibinfo{author}{Cesare \surnamestart Tinelli\surnameend}
  (\bibinfo{year}{2006}): \emph{\bibinfo{title}{Solving SAT and SAT Modulo
  Theories: From an Abstract Davis--Putnam--Logemann--Loveland Procedure to
  DPLL(T)}}.
\newblock {\sl \bibinfo{journal}{Journal of the ACM}}
  \bibinfo{volume}{53}(\bibinfo{number}{6}), pp. \bibinfo{pages}{937--977},
  \doi{10.1145/1217856.1217859}.

\bibitemdeclare{incollection}{Nieuwenhuis2001-har}
\bibitem{Nieuwenhuis2001-har}
\bibinfo{author}{Robert \surnamestart Nieuwenhuis\surnameend} \&
  \bibinfo{author}{Albert \surnamestart Rubio\surnameend}
  (\bibinfo{year}{2001}): \emph{\bibinfo{title}{{Paramodulation-Based Theorem
  Proving}}}.
\newblock In \bibinfo{editor}{Alan \surnamestart Robinson\surnameend} \&
  \bibinfo{editor}{Andrei \surnamestart Voronkov\surnameend}, editors: {\sl
  \bibinfo{booktitle}{Handbook of Automated Reasoning}}, \bibinfo{volume}{I},
  pp. \bibinfo{pages}{371--443}, \doi{10.1016/B978-044450813-3/50009-6}.

\bibitemdeclare{inproceedings}{rosen2015tip}
\bibitem{rosen2015tip}
\bibinfo{author}{Dan \surnamestart Ros{\'e}n\surnameend} \&
  \bibinfo{author}{Nicholas \surnamestart Smallbone\surnameend}
  (\bibinfo{year}{2015}): \emph{\bibinfo{title}{TIP: Tools for Inductive
  Provers}}.
\newblock In \bibinfo{editor}{Martin \surnamestart Davis\surnameend},
  \bibinfo{editor}{Ansgar \surnamestart Fehnker\surnameend},
  \bibinfo{editor}{Annabelle \surnamestart McIver\surnameend} \&
  \bibinfo{editor}{Andrei \surnamestart Voronkov\surnameend}, editors: {\sl
  \bibinfo{booktitle}{Logic for Programming, Artificial Intelligence, and
  Reasoning (LPAR)}}, \bibinfo{publisher}{Springer}, pp.
  \bibinfo{pages}{219--232}, \doi{10.1007/978-3-662-48899-7_16}.

\end{thebibliography}
\end{document}